\documentclass[11pt,reqno]{amsart}

\input{preamble}

\title{Risk measures based on weak optimal transport}

\author{Michael Kupper}
\address{Department of Mathematics and Statistics, University of Konstanz}
\email{kupper@uni-konstanz.de}

\author{Max Nendel}
\address{Center for Mathematical Economics, Bielefeld University}
\email{max.nendel@uni-bielefeld.de}
	
\author{Alessandro Sgarabottolo}
\address{Center for Mathematical Economics, Bielefeld University}
\email{alessandro.sgarabottolo@uni-bielefeld.de}
	
\thanks{Financial support through the Deutsche Forschungsgemeinschaft (DFG, German Research Foundation) -- SFB 1283/2 2021 -- 317210226 is gratefully acknowledged.}
	
	\date{\today}
	
\begin{document}	
	
	\begin{abstract}
		In this paper, we study convex risk measures with weak optimal transport penalties.\ In a first step, we show that these risk measures allow for an explicit representation via a nonlinear transform of the loss function.\ In a second step, we discuss computational aspects related to the nonlinear transform as well as approximations of the risk measures using, for example, neural networks.\ Our setup comprises a variety of examples, such as classical optimal transport penalties, parametric families of models, uncertainty on path spaces, moment constrains, and martingale constraints.\ In a last step, we show how to use the theoretical results for the numerical computation of worst-case losses in an insurance context and no-arbitrage prices of European contingent claims after quoted maturities in a model-free setting. \smallskip
		
		 \smallskip
		
		
	\end{abstract}

\maketitle

\section{Introduction}
A key component of financial modeling lies in the description of the distribution of underlying risk factors. Since these risk factors typically take values in high-dimensional vector spaces and exhibit complex dependence structures, there is a natural trade-off between simplicity or implementability of a probabilistic model on the one hand and the most accurate description of reality on the other. Since a perfect reflection of reality in a single probabilistic framework is almost never possible, financial modeling is almost always accompanied by model uncertainty.\ A mathematical framework that allows to include model uncertainty in the assessment of financial contracts is given by the theory of risk measures, cf.\ \cite{foellmer2016finance}.\ Under appropriate continuity conditions, every convex risk measure admits a robust representation $\rho(f) = \sup_{\nu}\int f\,\d\nu - \alpha(\nu)$, where the evaluation $\int f\,\d\nu$ of a financial position (or loss function) $f$ is performed among a family of probability measures, which are penalized according to their plausibility using a suitable penalty function $\alpha$.

In the past decade, a strand of literature has formed around risk measures, where the penalty term  $\alpha(\nu)$ is given as a minimal cost of transportation from a given reference measure $\mu$, cf. \cite{bartl2020computational}, or by a suitable function applied to the Wasserstein distance to the reference measure $\mu$, see \cite{nendel2022parametric} and \cite{bartl2021randomwalks,nendel2023wasserstein} for a dynamic setting. Here, the reference measure $\mu$ might be, for example, a model that is particularly attractive due to its computational simplicity, suggested by some expert, or purely data-driven, e.g., an empirical distribution derived from given data points, cf. \cite{obloj2021robust}.\ A similar idea is inherent to the discipline of distributionally robust optimization, where an additional optimization over (possibly rescaled versions of) such risk measures is performed, cf. \cite{gao2023distributionally,esfahani2018data,pflug2007ambiguity}. We also refer to \cite{bartl2021sensitivity,bartl2023sensitivity, bart2023sensitivityadapted} for an analysis of the sensitivity of robust optimization problems with respect to the degree of uncertainty.\ 
Moreover, robust versions of risk measures based on uncertainty sets are investigated in \cite{bernard2022robust, pesenti2020optimizing}.

The aim of this article is a theoretical and numerical study of convex risk measures based on weak optimal transport penalties.\ More precisely, we focus on risk measures, whose penalty functions measure the distance of a probability $\nu$ on a Borel space $Y$ to a reference probability $\mu$ on a Borel space $X$ by means of a weak optimal transport cost of the form
\[
\alpha(\nu) := \inf_{\kappa \in \ker(\mu, \nu)} \int_X C\big(x, \kappa(x)\big)\, \mu(\d x),
\]
where $\ker(\mu, \nu)$ contains all kernels $\kappa
\colon X\times\mathcal{P}(Y)\to [0,1]$ such that $\pi=\mu\otimes \kappa$ is a coupling between $\mu$ and $\nu$. 
The special choice $C\big(x, \kappa(x)\big) := \int_Y c(x, y) \kappa(x, \d y)$ with a cost function $c\colon X\times Y \to [0, \infty]$ results in a standard optimal transport problem as studied, e.g., in \cite{bartl2020computational,nendel2022parametric}.\ The weak optimal transport problem was initially introduced in \cite{MR3706606} and shortly afterwards in \cite{MR4028478}.\ For an overview on basic results of optimal transport theory under weak transportation costs and a collection of applications from different areas, we refer to \cite{MR4337709}. For instance, in stochastic optimization and mathematical finance, weak optimal transport appears in the context of martingale transport problems \cite{beiglboeck2013model}, causal transport problems \cite{MR3738324}, semimartingale transportation problems \cite{MR3127880}, and stability results for pricing and hedging \cite{MR4118990}.\ In our framework, the weak optimal transport cost allows to consider a wide spectrum of examples, which range from classical optimal transport penalizations, over parametric uncertainty and uncertainty on path spaces in the spirit of \cite{foellmer2022optimal} to penalties related to c\`adl\`ag martingales interpolating between the reference measure $\mu$ as initial marginal and a terminal marginal $\nu$, see Section \ref{sec:examples} for more details.

In a first step, we show that the risk measure
\[
\rho(f):=\sup_{\nu\in \cP(Y)} \bigg(\int_Y f(y)\,\nu(\d y)-\inf_{\kappa\in \ker(\mu,\nu)}\int_X C\big(x,\kappa(x)\big)\,\mu(\d x)\bigg)
\]
can be expressed in terms of a nonlinear transform of the loss function $f$ which is independent of the reference measure $\mu$. Theorem \ref{thm.main} shows that
\[
\rho(f) = \int_X f^C(x)\,\overline\mu(\d x),
\]
 where $f^C(x)=\sup_{\vartheta \in \cP(Y)} \big(\int f\, \d \vartheta-C(x,\vartheta)\big)$ is the so-called $C$-transform of $f$ and $\overline\mu$ is the extension of $\mu$ to the universal $\sigma$-algebra on $X$.

In a second step, we exploit the fact that the inner expression in the $C$-transform is measure affine and show that, in special yet relevant cases, the computation of the $C$-transform can be reduced to the extreme points.\ In particular we use a classical representation of the extreme points of sets of probability measures given in terms of generalized moment constraints, cf.\ \cite{weizsaecker1979irp,winkler1988extreme}.
This allows also for a variational formulation of $\rho(f)$ based on an optimization over Borel measurable functions taking values in the set of extreme points, see Proposition \ref{prop:selection}.

In the case of a Wasserstein or martingale Wasserstein penalty, this allows to use a neural network approximation in the spirit of \cite{nendel2022parametric}.\ The idea of using neural network approximations for the computation of risk measures based on classical optimal transport has recently been exploited in \cite{eckstein2021computation,eckstein2020robust}, see also \cite{de2020minmax,henry2019martingale} and the references therein for min-max methods.\ However, the approach contained in these works relies on duality results, which transform the original problem into the dual superhedging problem. In our approach, we use a neural network to model relevant transformations of the reference measure which are able to approximate the worst-case measure.\ Therefore, while the approach of \cite{eckstein2021computation,eckstein2020robust} provides an approximation from above, and gives as a byproduct the optimal superhedging strategy for the robust hedging problem, our approach provides an approximation from below, and gives as a byproduct the $C$-transform of the loss function.\ Note that, having the $C$-transform of the loss function at hand, allows for a fast re-evaluation of the risk measure as the reference model changes at the cost of a Monte-Carlo simulation.\ This makes the approach also appealing for risk management, where the possibility to promptly adjust the evaluation of a risk as the baseline model changes is of great advantage.

An elementary reformulation of the martingale constraint, which is a priori an infinite-dimensional constraint, allows to understand it as a single moment constraint, and allows to consider robust price bounds under no-arbitrage considerations in a model-free setting, see Section \ref{sec:mart.constr} and Section \ref{sec:model.free.pricing}.\ In Section \ref{sec:pricing.after}, we consider European claims with maturities that exceed the last quoted maturity, and model the absence of information after the quoted maturities by penalizing deviations from the last quoted distribution of the underlying assets with a martingale Wasserstein distance.\ Using the neural network approximation scheme, we  compute no-arbitrage bounds for the prices of options written on a possibly large number of underlying assets, see also Section \ref{sec:higherdim}.

The rest of the paper is organized as follows. Section \ref{sec:main} introduces the setup and contains the representation of the risk measure $\rho$ via the $C$-transform of the loss function $f$, see Theorem \ref{thm.main}.\ In Section \ref{sec:constraints}, we introduce a specific class of cost functions that allow for more explicit solutions, when the set of extreme points of the ambiguity set is known, see Theorem \ref{thm.reduction.extr}. In Section \ref{sec:examples}, we present a variety of examples that are covered by our approach and indicate possible applications.\ Finally, in Section \ref{sec.numerics}, we formulate our neural network approximation based on a general approximation result for Wasserstein penalties, see Proposition \ref{prop.dense.approx}.\ This approximation based on neural networks is then applied to several examples from finance and insurance.

\subsection*{Notation.} Given a topological space $\Omega$, we denote by $\cB(\Omega)$ its Borel $\sigma$-algebra, and by $\cP(\Omega)$ the set of all probability measures on $\cB(\Omega)$, endowed with the weak topology.\ Let $\overline \R:=\R\cup\{\pm \infty\}$.\ Then, we denote by ${\rm B}(\Omega)$ the space of all Borel measurable functions $\Omega\to \overline\R$ and by ${\rm B}_{\rm b}(\Omega)$ the space of all bounded Borel measurable functions $\Omega\to \R$.\ For a probability measure $\mu\in \cP(\Omega)$, let $\overline \sigma(\mu):=\cB(\Omega)\vee \mathscr N_\mu$, where $\mathscr N_\mu$ consists of all sets $N\in 2^\Omega$ with $N\subset B$ for some $B\in \cB(\Omega)$ with $\mu(B)=0$. Then, we denote by $\mathcal U(\Omega):=\bigcap_{\mu\in \cP(\Omega)}\overline \sigma(\mu)$ the $\sigma$-algebra of all universally measurable sets, and by $\overline \mu$ the extension of $\mu$ to $\mathcal U(\Omega)$ for all $\mu\in \cP(\Omega)$.
For a universally measurable function $f\colon \Omega\to \overline \R$ and $\mu\in \cP(\Omega)$, we define
\[
\overline \mu f:=\int_\Omega f(x)\, \overline\mu(\d x):= \int_\Omega f^+(x)\, \overline\mu(\d x)-\int_\Omega f^-(x)\, \overline \mu(\d x)\in \overline \R.
\]
Here and throughout, we use the convention $\infty-\infty=-\infty$. If $f\in B(\Omega)$, we simply write $\mu f$ and $\overline \mu f$ instead of $\int_\Omega f(x)\, \mu(\d x)$ and $\int_\Omega f(x)\, \overline \mu(\d x)$, respectively.\
For $\mu\in \cP(\Omega)$, we denote by $\mathscr L^1(\mu)$ the space of all $f\in {\rm B}(\Omega)$ with $\int_\Omega |f(x)|\, \mu(\d x)<\infty$, where we use the convention $|-\infty|=\infty$. In a similar way, we denote by $\mathscr L^1(\overline \mu)$ the space of all universally measurable functions $f\colon \Omega\to \overline \R$ with $\int_\Omega |f(x)|\, \overline \mu(\d x)<\infty$.

Let $(S,\mathscr S)$ be a measurable space.\ We denote the space of all $\cB(\Omega)$-$\mathscr S$-measurable functions $\Omega\to S$ by $\Bm(\Omega,S)$.\ For a probability measure $\mu\in \cP(\Omega)$ and a function $g\in \Bm(\Omega,S)$, we define the push-forward $g(\mu)\colon \mathscr S\to [0,1]$ by
\[
\big(g(\mu)\big)(A):=\mu(g\in A)\quad\text{for all }A\in \mathscr S.
\]

\section{Weak transport penalties} \label{sec:main}

Throughout, let $X,Y$ be two Borel spaces, cf.\ \cite[Definition 7.7]{bertsekas1978stochastic}.\ Since $Y$ is a Borel space, the set $\cP(Y)$ of all Borel probability measures on $Y$ is again a Borel space, cf.\ \cite[Corollary 7.25.1]{bertsekas1978stochastic}. We denote by $\ker(X,Y)$ the set of all Borel measurable stochastic kernels $\kappa$ from $X$ to $Y$, i.e., $\kappa\colon X\times \cB(Y)\to [0,1]$ with
\begin{enumerate}
\item[(i)] $\kappa(x)=\kappa(x,\, \cdot\;)\in \cP(Y)$ for all $x\in X$,
\item[(ii)] the map $X\to [0,1],\; x\mapsto \kappa(x,B)$ is Borel measurable for all $B\in \cB(Y)$.
\end{enumerate}
For $\mu\in \cP(X)$ and $\nu\in \cP(Y)$, we define $$\ker(\mu,\nu):=\big\{\kappa\in \ker(X,Y) \colon \mu\kappa=\nu\big\},$$ where $(\mu\kappa)(B):=\int_X\kappa(x,B)\,\mu(\d x)$ for all $B\in \cB(Y)$.

In the following, let $\mu\in \cP(X)$ be a reference measure and $C\colon X\times \cP(Y)\to [0,\infty]$ be a Borel measurable cost function.\ We assume that there exists some $\kappa_0\in \ker(X,Y)$ with $C\big(x,\kappa_0(x)\big)=0$ for $\mu$-almost all $x\in X$. For $f\in \Bm(Y)$, let
\begin{equation}\label{def:RM}
\rho(f):=\sup_{\nu\in \cP(Y)} \bigg(\int_Y f(y)\,\nu(\d y)-\inf_{\kappa\in \ker(\mu,\nu)}\int_X C\big(x,\kappa(x)\big)\,\mu(\d x)\bigg),
\end{equation}
where we use the convention $\inf\emptyset =\infty$.\ By assumption, $$
\int_Y f(y)\,(\mu\kappa_0)(\d y)\leq \rho(f)\quad\text{for all }f\in \Bm(Y),
$$
i.e., the kernel $\kappa_0$ can be interpreted as a transform of states in $X$ into states in $Y$, which is not penalized, or, in other words, $\rho$ incorporates a risk loading w.r.t.\ the transformed reference measure $\mu\kappa_0$.\ Observe that the functional $\rho\colon \Bm(Y)\to \overline \R$ defines a \textit{monetary risk measure}, i.e.,
\begin{enumerate}
 \item[(i)] $\rho(f)\leq \rho(g)$ for all $f,g\in \Bm(Y)$ with $f\leq g$,\footnote{For $f,g\in \Bm(X)$, we write $f\leq g$ if $f(x)\leq g(x)$ for all $x\in X$.}
 \item[(ii)] $\rho(0)=0$ and $\rho(f+m)=\rho(f)+m$ for all $f\in \Bm(Y)$ and $m\in \R$,
\end{enumerate}
and that the weak optimal transportation cost $\alpha(\nu):=\inf_{\kappa\in \ker(\mu,\nu)}\int_X C\big(x,\kappa(x)\big)\,\mu(\d x)$ for each $\nu\in \cP(Y)$ is a penalty function of the risk measure $\rho$.\ We refer to \cite[Chapter 4]{foellmer2016finance} for a detailed discussion on monetary risk measures.

We define the \textit{$C$-transform} of $f\in \Bm(Y)$ by
\[
f^C(x):=\sup_{\vartheta \in \cP(Y)} \big(\vartheta f-C(x,\vartheta)\big)\in \overline\R\quad\mbox{for all }x\in X.
\]
Since $C\big(x,\kappa_0(x)\big)=0$ for $\mu$-almost all $x\in X$, it follows that $\kappa_0(x)f \leq f^C(x)$ for $\mu$-almost all $x\in X$. Moreover, $f^C$ is upper semianalytic and therefore universally measurable, cf.\ \cite[Proposition 7.29]{bertsekas1978stochastic} and \cite[Proposition 7.47]{bertsekas1978stochastic}.
In the special case, where $C(x, \delta_y) := c(x,y)$ for all $x\in X$ and $y\in Y$ with some Borel measurable function $c\colon X\times Y\to [0,\infty]$, and 
$C(x, \vartheta) := \infty$ for all $x\in X$ if $\vartheta\in \cP(Y)$ is not a Dirac measure, the $C$-transform reduces to the well-known \textit{$c$-transform} $f^c(x)=\sup_{y\in Y} \big(f(y)-c(x,y)\big)$ for all $x\in X$. 

In the following, we restrict the risk measure $\rho$ to the set
$$
\cL:=\bigg\{f\in {\rm B}(Y)\colon -\infty<\int_Y f(y)\,(\mu\kappa_0)(\d y)\text{ and }\int_X f^C(x)\, \overline\mu(\d x)<\infty\bigg\}.
$$
Observe that ${\rm B}_{\rm b}(Y)\subset \cL$. In particular, $\cL$ contains all constant functions.\ The following characterization of
the risk measure \eqref{def:RM} is useful for its numerical computation in Section~\ref{sec.numerics}.  
\begin{theorem}\label{thm.main}
The set $\cL$ is convex and $\rho\colon \cL\to \R$ defines a convex risk measure, i.e.,
$$
\rho\big(\lambda f+(1-\lambda)g\big)\leq \lambda \rho(f)+(1-\lambda)\rho(g) \quad\text{for all }f,g\in \cL\text{ and }\lambda\in [0,1].
$$
Moreover, for all $f\in \cL$,
\begin{equation}\label{eq.main}
 \rho(f)=\sup_{\kappa\in \ker(X,Y)}\int_X \kappa(x) f-C\big(x,\kappa(x)\big)\,\mu(\d x)=\int_X f^C(x)\,\overline\mu(\d x).
\end{equation}
\end{theorem}

\begin{proof}
 Let $f\in \cL$.\ Since $\inf_{\kappa\in \ker(\mu,\nu)}\int_X C(x,\kappa(x))\,\mu(\d x)=\infty$ whenever $\ker(\mu,\nu)=\emptyset$ and $\mu\kappa=\nu$ for all $\kappa\in \ker(\mu,\nu)$, we find that
 \begin{align*}
 \rho (f)&=\sup_{\nu\in \cP(Y)}\sup_{\kappa\in \ker(\mu,\nu)}\left( \mu \kappa f- \int_X C\big(x,\kappa(x)\big)\,\mu(\d x)\right)\\
 &=\sup_{\kappa\in \ker(X,Y)}\left( \mu \kappa f- \int_X C\big(x,\kappa(x)\big)\,\mu(\d x)\right)\\
 &= \sup_{\kappa\in \ker(X,Y)} \int_X \kappa(x) f-C\big(x,\kappa(x)\big)\,\mu(\d x).
 \end{align*}
 Hence,
 \begin{equation}\label{eq.integrability}
 -\infty<\int_Y f(y)\,(\mu\kappa_0)(\d y)= \mu \kappa_0f \leq \rho(f)\quad \text{and}\quad\rho(f)\leq \int_X f^C(x)\,\overline\mu(\d x)<\infty.
 \end{equation}
 Now, let $\eps>0$ and $\vartheta\colon X\to \cP(Y)$ be a universally measurable $\eps$-selection of $f^C$, i.e.,
 \[
 f^C(x)\leq \vartheta(x)f-C\big(x,\vartheta(x)\big)+\eps\quad\text{for all }x\in X,
 \]
 cf.\ \cite[Proposition 7.50]{bertsekas1978stochastic}. Then, by \cite[Lemma 7.28]{bertsekas1978stochastic}, there exists some $\kappa\in \ker(X,Y)$ with $\kappa(x)=\vartheta(x)$ for $\mu$-almost all $x\in X$. We thus obtain that
 \begin{align*}
 \int_X f^C(x)\, \overline\mu(x)&\leq \int_X \vartheta(x)f-C\big(x,\vartheta(x)\big)\, \overline\mu(x)+\eps\\
 &=\int_X \kappa(x)f-C\big(x,\kappa(x)\big)\, \mu(x)+\eps\leq \rho(f)+\eps.
 \end{align*}
 Now, \eqref{eq.main} follows by passing to the limit $\eps\downarrow 0$.
 
 It remains to show that $\cL$ is convex and that $\rho\colon \cL\to \R$ is convex.\ To that end, let $f,g\in \cL$ and $\lambda\in [0,1]$.\ By \eqref{eq.integrability}, it follows that $f,g\in \mathscr L^1(\mu\kappa_0)$ and $f^C,g^C\in \mathscr L^1(\overline\mu)$.\ Since $\mathscr L^1(\mu\kappa_0)$ and $\mathscr L^1(\overline\mu)$ are vector spaces, $\lambda f+(1-\lambda)g\in \mathscr L^1(\mu\kappa_0)$ and $\lambda f^C+(1-\lambda)g^C\in \mathscr L^1(\overline\mu)$.\ Hence, in view of \eqref{eq.main}, both the convexity of $\cL$ and $\rho\colon \cL\to \R$ follow as soon as we have shown that
 \begin{equation}\label{eq.conv.cconj}
 \big(\lambda f+(1-\lambda)g\big)^C\leq \lambda f^C+(1-\lambda)g^C\quad\overline\mu\text{-almost surely}.
 \end{equation}
 Since $f^C,g^C\in \mathscr L^1(\overline\mu)$, there exists a set $A\in \cU(X)$ with $\overline\mu(A)=1$ as well as
 \[
 f^C(x)<\infty \quad\text{and}\quad g^C(x)<\infty \quad\text{for all }x\in A.
 \]
 Let $x\in A$ and $\vartheta\in \cP(Y)$ with $c(x,\vartheta)<\infty$. Then, $\vartheta f<\infty$ and $\vartheta g<\infty$. Hence,
 $$
    \vartheta \big(\lambda f+(1-\lambda)g\big)-C(x,\vartheta)=  \big(\lambda \vartheta f+(1-\lambda) \vartheta g\big)-C(x,\vartheta)\leq \lambda f^C(x)+(1-\lambda)g^C(x)
 $$
 for all $x\in A$.\ Taking the supremum over all $\vartheta\in \cP(Y)$ with $c(x,\vartheta)<\infty$, it follows that
 \[
  \big(\lambda f+(1-\lambda)g\big)^C=\sup_{\substack{\vartheta\in \cP(Y) \\ C(x,\vartheta)<\infty}}\Big(\vartheta \big(\lambda f+(1-\lambda)g\big)-C(x,\vartheta)\Big)\leq \lambda f^C(x)+(1-\lambda)g^C(x)
 \]
 for all $x\in A$.\ We have therefore shown the validity of \eqref{eq.conv.cconj}, and the claim follows. 
\end{proof}

\section{Optimal transport costs with additional constraints}\label{sec:constraints}
In this section, we consider weak transport costs with constraints given in terms of a convex and Borel measurable set $\cM \subset \cP(Y)$ and a family $(\tau_x)_{x\in X}$ of bijective maps $Y\to Y$.\ We assume that both $(x,y)\mapsto \tau_x(y)$ and $(x,y)\mapsto \tau_x^{-1}(y)$ are Borel measurable maps $X\times Y\to Y$.\ Then, by \cite[Lemma 7.12 and Proposition 7.25]{bertsekas1978stochastic}, the maps $(x,\vartheta)\mapsto \tau_x(\vartheta)$ and $(x,\vartheta)\mapsto \tau_x^{-1}(\vartheta)$ are Borel measurable $X\times \cP(Y)\to \cP(Y)$.\ Moreover, for all $x\in X$, the map $\cP(Y)\to \cP(Y),\, \vartheta\mapsto \tau_x(\vartheta)$ is bijective with inverse  $\cP(Y)\to \cP(Y),\, \vartheta\mapsto \tau_x^{-1}(\vartheta)$.

Throughout this section, we consider a weak optimal transport cost of the form
\begin{equation} \label{eq.transport.cost}
C(x, \vartheta) = \begin{cases}\int_Y c(x,y)\,\vartheta(\d y), & \text{if } \tau_x^{-1}(\vartheta)\in \cM,
\\
\infty, & \text{otherwise,}
\end{cases}
\end{equation}
where $c\colon X \times Y \to [0, \infty]$ is a Borel measurable cost function.\ Since $c$, $\cM$, and the map $X\times \cP(Y)\to \cP(Y),\, (x,\vartheta)\mapsto \tau_x^{-1}(\vartheta)$ are Borel measurable, it follows that the cost function $C\colon X\times \cP(Y)\to [0,\infty]$ is Borel measurable.\ Again, we assume that there exists some $\kappa_0\in \ker(X,Y)$ with $C\big(x,\kappa_0(x)\big)=0$ for $\mu$-almost all $x\in X$, i.e., $\kappa_0(x)\tau_x^{-1}\in\mathcal{M}$ and $\int_Y c(x,y)\,\kappa_0(x,\d y)=0$ for $\mu$-almost all $x\in X$.\ Since the map $\cP(Y)\to \cP(Y),\, \vartheta\mapsto \tau_x(\vartheta)$ is bijective with inverse $\cP(Y)\to \cP(Y),\, \vartheta\mapsto \tau_x^{-1}(\vartheta)$, the $C$-transform of $f\in B(Y)$ is given by
$$
f^C(x)=\sup_{\vartheta\in \mathcal{M}}\Big(\tau_x(\vartheta)f-C\big(x,\tau_x(\vartheta)\big)\Big)=\sup_{\vartheta\in \mathcal{M}}\int_Y \tau_x\big(f(y)-c(x,y)\big)\,\vartheta(\d y),
$$
where, for $g\in \Bm(X\times Y)$, we use the notation
$$
\tau_xg(x,y):=g\big(x,\tau_x(y)\big)\quad\text{for all }x\in X\text{ and }y\in Y.
$$
 The following remark shows how, in the unconstrained case, this setup reduces to a penalty function, which is given in terms of a classical optimal transport problem.\ For additional examples of penalty functions and the related risk measures covered by this setup, we refer to Section \ref{sec:examples} below.

\begin{remark}\label{rem:transport}
In the unconstrained case, $\cM = \cP(Y)$ and $\tau_x=\id_Y$ for all $x\in X$, it follows that $C(x, \vartheta) = \int_Y c(x,y)\,\vartheta(\d y)$, and the weak optimal transport problem
$$\inf_{\kappa \in \ker(\mu, \nu)} \int_X C\big(x, \kappa(x)\big)\,\mu(\d x)$$
is equivalent to the classical optimal transport problem
$$\inf_{\pi \in \cpl(\mu, \nu)} \int_{X \times Y} c(x, y)\,\pi(\d x, \d y),$$
where $\cpl(\mu, \nu)$ is the set of couplings between $\mu$ and $\nu$, i.e., the set of all probability measures $\pi \in \cP(X \times Y)$ with $\prj_X(\pi) = \mu$ and $\prj_Y(\pi) = \nu$.\
In fact, on the one hand, for every $\kappa \in \ker(\mu, \nu)$, the condition $\mu \kappa = \nu$ implies that the measure $\pi \in \cP(X \times Y)$
given by $\pi(A \times B) := \int_A \kappa(x, B)\,\mu(\d x)$ for all $A\in \cB(X)$ and $B\in \cB(Y)$ is a coupling between $\mu$ and $\nu$. On the other hand, since $X$ and $Y$ are Borel spaces, for each $\pi \in \cpl(\mu, \nu)$, there exists a Borel measurable stochastic kernel $\kappa \colon X\times \cB(Y) \to [0,1]$ with 
$\pi(A\times B) = \int_A \kappa(x,B)\,\mu(\d x)$ for all $A \in \cB(X)$ and $B \in \cB(Y)$, cf.\ \cite[Corollary 7.27.2]{bertsekas1978stochastic}, which is an element of $\ker(\mu, \nu)$ since $\int_X \kappa(x,B)\,\mu(\d x) = \pi(X \times B) = \nu(B)$ for all $B \in \cB(Y)$.
\end{remark}

In the following, we focus on the general case, where $\cM$ is not necessarily $\cP(Y)$.\ Recall that a probability measure $\theta \in \cM$ is an extreme point of $\cM$, if $\theta = \lambda \theta_1 + (1 - \lambda) \theta_2$ for $\theta_1, \theta_2 \in \cM$ and $\lambda \in (0,1)$ implies $\theta = \theta_1 = \theta_2$.\ We denote by $\ex \cM$ the set of all extreme points of $\cM$.\ The set $\ex \cM$ is endowed with the smallest $\sigma$-algebra $\Sigma(\ex \cM)$ on $\ex \cM$ such that $\ex \cM \to [0,1]$, $\theta \mapsto \theta(B)$ is measurable for all $B \in \cB(Y)$.\ By \cite[Proposition 7.25]{bertsekas1978stochastic}, $\Sigma(\ex \cM)$ coincides with the trace $\sigma$-algebra of $\cB(\cP(Y))$ on $\ex\cM$, which is the same as the Borel $\sigma$-algebra of the subspace topology of the weak topology on $\ex\cM$.

The next result relies on a characterization given in \cite{weizsaecker1979irp,winkler1988extreme}. The set $\cM$ is said to satisfy the \textit{integral representation property (IRP)} if, for 
each $\vartheta \in \cM$, there exists a probability measure $p \colon \Sigma(\ex \cM) \to [0,1]$ such that
\[
\vartheta(B) = \int_{\ex \cM} \theta(B)\,p(\d \theta) \quad \text{for all } B \in \cB(Y).
\]

For the sake of illustration, we provide three examples for sets that satisfy the IRP, and refer to \cite{weizsaecker1979irp} for a more profound study of this property.

\begin{example}\ \label{ex.irp}
\begin{enumerate}[a)]
 \item Let $\cM \subset \cP(Y)$ be nonempty, convex, and weakly closed.\ Then, $\cM$ has the IRP, cf.\ \cite[Corollary 1]{weizsaecker1979irp}.
 \item Let $\phi\colon Y\to \R$ be Borel measurable. Then,
 $$\cM:=\big\{\vartheta \in \cP(Y)\colon \phi\in \mathscr{L}^1(\vartheta)\big\}$$ satisfies the IRP, cf.\  \cite[Corollary 3]{weizsaecker1979irp}. \label{ex.irpb}
 \item \label{ex.irp.moments} For each $m\in \N$, let $\phi_m\colon Y\to \R$ be Borel measurable and $a_m\in \R$. Then, the set $$\cM:=\bigg\{\vartheta \in \cP(Y)\colon \phi_m\in \mathscr L^1(\vartheta) \text{ and }\int_Y \phi_m(y)\,\vartheta(\d y)\leq a_m \text{ for all }m\in \N\bigg\}$$
 has the IRP, cf.\ \cite[Corollary 3]{weizsaecker1979irp}.
\end{enumerate}
\end{example}

The IRP is of particular interest in our framework since, for every \textit{measure affine} functional on $\cM$, the supremum with respect to probability measures in $\cM$ can be restricted to the supremum over $\ex \cM$, see \cite[Theorem 3.2 and Proposition 3.1]{winkler1988extreme}.\
In fact, for $\vartheta\in\cP(Y)$ and every probability measure $p\colon \Sigma(\ex \cM)\to [0,1]$ with
$\vartheta(B)= \int_{\ex \cM} \theta(B)\,p(\d \theta)$ for all $B\in \cB(Y)$, measure-theoretic induction yields the \textit{barycentrical formula}\footnote{Recall that $\vartheta g:=\int_Y g^+(y)\,  \vartheta(\d y)- \int_Y g^-(y)\, \vartheta(\d y)$ with the convention $\infty-\infty=-\infty$.}
$$
\vartheta g=\int_{\ex \cM} \theta g\,p(\d \theta)\quad \text{for all } g\in B(Y).
$$
In particular, if $\cM\subset \cP(Y)$ satisfies the IRP and $g\in \Bm(Y)$, then $\vartheta g\le \sup_{\theta\in \ex \cM} \theta g$ for all $\vartheta\in\cM$.\ Taking the supremum over all $\vartheta\in\cM$, we get
 \begin{equation}\label{eq:barycenter}
    \sup_{\vartheta \in \cM} \int_Y g(y)\,\vartheta(\d y) = \sup_{\theta \in \ex \cM} \int_Y g(y)\,\theta(\d y).
\end{equation}
In combination with Theorem~\ref{thm.main}, we obtain the following result, which makes the computation of the $C$-conjugate $f^C$ more tractable.

\begin{theorem} \label{thm.reduction.extr}
Assume that the cost function $C$ is of the form \eqref{eq.transport.cost}, and let $\ex \cM \subset \Theta \subset\cM$. If $\cM$ satisfies the IRP, then
\[
\rho(f) = \int_X \sup_{\theta \in \Theta} \bigg(\int_Y\tau_x\big(f(y) - c(x,y) \big)\,\theta(\d y)\bigg)\,\overline{\mu}(\d x)\quad\text{for all }f\in \cL.
\]
\end{theorem}

\begin{proof}
Let $f\in\cL$.\ Using equation~\eqref{eq:barycenter}, it follows that
 \begin{align*}
    f^C(x) & = \sup_{\vartheta \in \cM} \Big(\tau_x (\vartheta) f - C\big(x, \tau_x(\vartheta) \big)\Big) \ge \sup_{\theta \in \Theta}  \Big( \tau_x(\theta) f - C\big(x, \tau_x(\vartheta)\big) \Big)\\
    &\ge \sup_{\vartheta \in \ex \cM} \Big(\tau_x(\vartheta) f - C\big(x, \tau_x(\vartheta) \big)\Big)
     = \sup_{\vartheta \in \ex \cM}  \int_Y \tau_x\big(f(y) - c(x, y)\big)\,\vartheta(\d y) \\
     &=\sup_{\vartheta \in \cM} \int_Y \tau_x\big(f(y) - c(x, y)\big)\,\vartheta(\d y)=f^C(x).
    \end{align*}
  Hence,
  $$
  f^C(x)=\sup_{\theta \in \Theta}  \Big( \tau_x(\theta )f - C\big(x, \tau_x(\theta) \big) \Big)=\sup_{\vartheta \in \Theta}  \int_Y \tau_x\big(f(y) - c(x, y)\big)\,\theta(\d y) \quad\text{for all }x\in X. 
  $$
  Since $f^C$ is universally measurable, the claim follows from Theorem \ref{thm.main}.
 \end{proof}

We conclude this section with the following proposition, which is another consequence of Theorem \ref{thm.main} and is closely connected to \cite[Theorem 2.4]{bartl2020computational}.

\begin{proposition}\label{prop.reduction.extr}
Suppose that $\cM$ contains all Dirac measures on $Y$ and that the cost function $C$ is of the form \eqref{eq.transport.cost}.\ Then, 
$$
\rho(f)  =  \int_X \sup_{y \in Y} \big(f(y) - c(x,y) \big)\,\overline\mu(\d x)=\int_X f^c(x)\, \overline\mu(\d x)\quad\text{for all }f\in \cL.
$$
 \end{proposition}

 \begin{proof}
  By \cite[Proposition 7.17]{bertsekas1978stochastic}, every Borel probability measure on a Borel space is regular. Hence, the set of all probability measures on a Borel space has the Dirac measures as extreme points, see \cite[Theorem 11.1]{topsoe1970topology}. Since $\cM$ contains all the Dirac measures on $Y$, it follows that
   \begin{align*}
    f^C(x) & = \sup_{\vartheta \in \cM} \Big(\tau_x(\vartheta) f - C\big(x, \tau_x(\vartheta)\big)\Big) \geq \sup_{y \in Y} \tau_x\big(f(y) - c(x,y)\big)\\
    &=\sup_{\vartheta \in \cP(Y)} \big(\tau_x(\vartheta) f - C\big(x, \tau_x(\vartheta) \big)\Big) \geq f^C(x).
   \end{align*}
   We have therefore shown that 
   $$
   f^C(x)=\sup_{y \in Y} \tau_x\big(f(y) - c(x,y)\big)=\sup_{y \in Y} \big(f(y) - c(x,y)\big)=f^c(x)\quad\text{for all }x\in X,
   $$
   where the second equality follows from the fact that $\tau_x\colon Y\to Y$ is bijective for all $x\in X$, and the claim follows again from Theorem~\ref{thm.main}.
\end{proof}

\subsection{Generalized moment constraints} \label{sec.moments}
Theorem \ref{thm.reduction.extr} is particularly useful when the set of extreme points of $\cM$ has a simple and explicit representation.\ This is, for example, the case if we impose \textit{generalized moment constraints} in the sense that we consider the set 
\begin{equation} \label{eq.moments}
\cM = \Big\{ \vartheta \in \cP(Y)\,:\, f_i\in \mathscr L^1(\vartheta) 
\text{ and } \int_Y f_i(y)\,\vartheta(\d y) \le c_i \text{ for } 1 \le i \le n \Big\}
\end{equation}
with $n\in \N_0$, Borel measurable functions $f_1, \dots, f_n\colon Y\to \overline \R$, and $c_1, \dots, c_n\in \R$.\ As already remarked in Example \ref{ex.irp} \ref{ex.irp.moments}), the set $\cM$ satisfies the IRP.
As remarked in the proof of Proposition \ref{prop.reduction.extr}, the set of all Borel probability measures on $Y$ has the Dirac measures as extreme points.\ This allows us to use the following theorem by Winkler, see \cite[Theorem 2.1]{winkler1988extreme}, which we report here in the special case $\cP=\cP(Y)$.

\begin{theorem}[Winkler, 1988] \label{thm:winkler}
    Let $\cM \subset \cP(Y)$ be defined as in \eqref{eq.moments}.\ Then, the set $\cM$ is convex and
    \begin{align*}
    \ex \cM \subset  \Big\{ &\vartheta \in \cM\colon \vartheta = \sum_{i = 1}^m a_i \delta_{y_i} \text{ with } y_i \in Y,\, a_i > 0,\, \sum_{i=1}^m a_i = 1,\, 1 \le m \le n+1,\\
	&\text{and the vectors } (f_1(y_i),\dots, f_n(y_i), 1),\, 1 \le i \le m, \text{ are linearly independent} \,\Big\}.
\end{align*}
    Equality of sets holds if the moment conditions in \eqref{eq.moments} are given by equalities.
\end{theorem}

 The following proposition, which is based on Theorem \ref{thm:winkler}, is a helpful tool for the numerical computation of the risk measure $\rho$ using neural networks.

\begin{proposition} \label{prop:selection}
 Assume that the cost function $C$ is of the form \eqref{eq.transport.cost} with $\cM$ given by \eqref{eq.moments}.\ Then, 
\begin{align}
    \rho(f) &=\sup_{(a,y)\in \Bm(X,\Theta)}\int_X \sum_{i=1}^{n+1}a_i(x)\tau_x\Big(f\big(y_i(x)\big)-c\big(x,y_i(x)\big)\Big)\,\mu(\d x)\notag \\
    & \label{eq.prop.moments} = \int_X \sup_{(a,y)\in \Theta} \sum_{i=1}^{n+1} a_i\tau_x\big(f(y_i)-c(x,y_i)\big), \overline\mu(\d x)\quad\text{for all }f\in \cL,
\end{align} 
 where 
 \[
 \Theta:=\bigg\{ (a,y)\in [0,1]^{n+1}\times Y^{n+1}\colon \sum_{i = 1}^{n+1} a_i \delta_{y_i} \in \cM \bigg\}.
 \]
\end{proposition}

\begin{proof}
 By Theorem \ref{thm.reduction.extr}, it follows that
 \[
 \rho(f)=\int_X\sup_{(a,y)\in \Theta} \sum_{i=1}^{n+1}a_i\tau_x\big(f(y_i)-c(x,y_i)\big)\, \overline\mu(\d x)\quad\text{for all }f\in \cL.
 \]
 In particular, $\geq$ holds in \eqref{eq.prop.moments}.\ Since the map from $[0,1]\times Y$ to the set of all subprobability measures on $\cB(Y)$, given by $(a,y)\mapsto a\delta_y$, is continuous and $\cM$ is weakly closed, it follows that $\Theta\subset [0,1]^{n+1}\times Y^{n+1}$ is closed.\ Hence, by \cite[Lemma 7.28 and Proposition 7.50]{bertsekas1978stochastic}, for all $f\in \cL$ and $\eps>0$, there exists a measurable map $(a,y)\colon X\to \Theta$ with
 \[
 \rho(f)\leq \int_X\sum_{i=1}^{n+1}a_i(x)\tau_x\Big(f\big(y_i(x)\big)-c\big(x,y_i(x)\big)\Big)\,\mu(\d x)+\eps.
 \]
 Taking the supremum over all measurable maps $X\to \Theta$ and letting $\eps\downarrow 0$, the claim follows. 
\end{proof}

 While Theorem \ref{thm.main} shows that the risk measure $\rho$ can be computed by solving a pointwise optimization that is independent of the reference measure $\mu$, Proposition \ref{prop:selection} suggests a variational approach, where the global behaviour of the loss function $f$, the cost $c$, and the reference measure $\mu$ are considered simultaneously.\ Since the maximization involves an integral with respect to the reference distribution, this optimization aligns with the framework of stochastic gradient descent methods, where tools such as neural networks typically perform very well.\ In Section \ref{sec.numerics} below, we will show how a neural network approximation can be formally justified in this context.

We conclude this section with the following corollary for a general set $\cM$ containing all Dirac measures.

\begin{corollary}\label{cor:selection}
 Suppose that $\cM$ contains all Dirac measures on $Y$ and that the cost function $C$ is of the form \eqref{eq.transport.cost}.\ Then, 
$$
\rho(f)  =  \sup_{y\in \Bm(X,Y)}\int_X f\big(y(x)\big) - c\big(x,y(x)\big)\,\mu(\d x)\quad\text{for all }f\in \cL.
$$
\end{corollary}

\begin{proof}
 By Proposition \ref{prop.reduction.extr}, it follows that
 \[
 \rho(f)  = \int_X\sup_{y\in Y}\big(f(y)-c(x,y)\big)\,\overline\mu(\d y).
 \]
 Hence, using Proposition \ref{prop:selection} in the unconstrained case, the claim follows. 
\end{proof}

\section{Examples and applications}\label{sec:examples}
In this section, we collect several examples that clarify the use of weak transport penalties, and discuss potential applications for the related risk measure.

\subsection{Parametric uncertainty} \label{sec:parametric.unc}
In this subsection, we present a setup that allows to incorporate special forms of parametric uncertainty via a weak optimal transport penalization.\ To that end, let $\mu\in\cP(X)$ and $\Theta$ be a $\sigma$-compact\footnote{Recall that a topological space $\Theta$ is $\sigma$-compact if there exists a sequence of compact subsets $(\Theta_n)_{n\in \N}$ of $\Theta$ with $\Theta=\bigcup_{n\in \N}\Theta_n$.} Hausdorff topological space, endowed with the Borel $\sigma$-algebra.\ Moreover, we assume that there exists a continuous and injective map $\gamma\colon \Theta\to \cP(Y),\, \theta \mapsto \gamma_\theta$, and define $\cM:=\gamma(\Theta)=\{\gamma_\theta\colon \theta\in \Theta\}$. Since $\gamma$ is continuous and $\Theta$ is $\sigma$-compact, it follows that $\cM$ is Borel measurable and that the inverse $\gamma^{-1}\colon \cM\to \Theta$ is Borel measurable.\footnote{This follows directly from the fact that the image of a compact set under a continuous map is compact and therefore $(\gamma^{-1})^{-1}(A)=\gamma(A)$ is a countable union of compacts for every closed subset $A\subset \Theta$.\ In particular, $\cM=\gamma(\Theta)$ is a countable union of compacts.}\ We consider a Borel measurable cost function $c\colon X\times \Theta\to [0,\infty]$, and assume that there exists a measurable map $\theta_0\colon X\to \Theta$ such that $c\big(x,\theta_0(x)\big)=0$ for $\mu$-almost all $x\in X$.\ For $x\in X$ and $\vartheta\in \cP(Y)$, we define
\[
C(x,\vartheta):=\begin{cases}
c\big(x,\gamma^{-1}(\vartheta)\big),&\text{if }\vartheta\in \cM,\\
\infty,&\text{otherwise}.
\end{cases}
\]
Since $c$, $\cM$, and $\gamma^{-1}$ are Borel measurable, it follows that $C\colon X\times \cP(Y)\to [0,\infty]$ is Borel measurable.\
Then, by Theorem \ref{thm.main},
$$
\rho(f)=\int_{X} \sup_{\theta\in \Theta} \big(\gamma_\theta f-c(x,\theta)\big)\,\overline\mu(\d x)\quad\text{for all }f\in \cL.
$$
Then, again by Theorem \ref{thm.main},
$$
\rho(f)=\sup_{\theta\in \Bm(X,\Theta)}\int_{X} \Big(\gamma_{\theta(x)} f-c\big(x,\theta(x)\big)\Big)\,\mu(\d x)\quad\text{for all }f\in \cL.
$$

\begin{example}
 Let $X=Y=\R$, $(\Omega,\cF,\P)$ be a probability space, and $\xi\colon \Omega\to \R$ be a random variable with $\E_\P(\xi)=0$.\ For $\theta=(m,\sigma)\in \Theta:=\R\times [0,\infty)$, let
 \[
 \gamma_\theta:=\P\circ (m+\sigma \xi)^{-1}.
 \]
Then, the map $\gamma\colon \Theta\to \cP(\R),\,\theta\mapsto \gamma_\theta$ is continuous, and we define the cost
\[
c(x,\theta):=\phi\big(|m-x|^2\big)+\psi\big(|\sigma|^2\big)\quad\text{for }x\in X \text{ and }\theta=(m,\sigma)\in \Theta
\]
with nondecreasing cost functions $\phi,\psi\colon [0,\infty)\to [0,\infty]$ with $\phi(0)=\psi(0)=0$.\ Defining $\theta_0(x):=(x,0)$ for all $x\in \R$, it follows that $c\big(x,\theta_0(x)\big)=0$ for all $x\in \R$.
\end{example}

\subsection{Wasserstein uncertainty} \label{sec:wass.unc}
In this section, we consider penalties, given in terms of Wasserstein distances. To that end, let $Y=X$ be a separable Banach space and  
$\cP_p(X)$ denote the set of all probability measures $\nu\in\cP(X)$ with finite moment $\int_X \|x\|^p\,\nu(\d x) < \infty$ of order $p\in[1,\infty)$. We assume that $\mu \in \cP_p(X)$ and consider the cost function
$$
C(x, \vartheta) := \int_X \|y - x\|^p\,\vartheta(\d y)\quad\text{for all }x\in X\text{ and }\vartheta \in \cP(X).
$$
Then, the cost $C$ is of the form \eqref{eq.transport.cost} with $c(x,y)=\|x-y\|^p$ and $\cM=\cP(X)$ and $\tau_x=\id_X$ for all $x\in X$.\
Moreover, $\kappa_0(x):=\delta_x$ satisfies $C(x,\kappa_0(x))=0$ for all $x\in X$ and we end up with the risk measure 
\begin{align*}
    \rho(f) & = \sup_{\nu \in \cP(X)} \bigg( \int_X f(y)\,\nu(\d y) - \inf_{\kappa \in \ker(\mu, \nu)} \int_X C\big((x, \kappa(x)\big)\,\mu(\d x)\bigg) \\
        & = \sup_{\nu \in \cP_p(X)} \bigg( \int_X f(y)\,\nu(\d y) - \cW_p(\mu, \nu)^p \bigg)\quad\text{for all }f\in \Bm(X),
\end{align*}
where $\cW_p(\mu, \nu) := \inf_{\pi \in \cpl(\mu, \nu)} \big( \int_X \|y - x\|^p\,\pi(\d x, \d y)\big)^{1/p}$ is the Wasserstein distance of order $p$ between $\mu$ and $\nu$.\ Different aspects of risk measures based on transport distances have recently been studied in \cite{bartl2021sensitivity,bartl2020computational,nendel2022parametric}.\ We also refer \cite{bartl2021randomwalks,bart2023sensitivityadapted, nendel2023wasserstein} for dynamic versions of these risk measures and to \cite{gao2023distributionally,esfahani2018data,pflug2007ambiguity} for applications in the context of distributionally robust optimization.\ Since $\cP_p(X)$ contains all Dirac measures, by Proposition~\ref{prop.reduction.extr}, 
\[
\rho(f) = \int_X \sup_{y \in X}\big(f(y) - \|x-y\|^p\big) \,\overline\mu(\d x)= \int_X \sup_{y \in X}\big(f(x+y) - \|y\|^p\big) \,\overline\mu(\d x)
\]
for all $f\in \cL$.

Now, let $f\in \Bm(X)$ and assume that there exist constants $M\geq0$ and $c\in [0,1)$ such that
\begin{equation}\label{eq.growthf}
|f(x)|\leq M+c \|x\|^p\quad \text{for all }x\in X.
\end{equation}
Since $\mu\in \cP_p(X)$, it follows that $f\in \mathscr L^1(\mu)$.\ Using the convexity of the map $[0,\infty)\to [0,\infty), \, v\mapsto v^p$, it follows that
\[
|f(x+y)|\leq M+\frac{c}{\lambda^{p-1}} \|x\|^p+ \frac{c}{(1-\lambda)^{p-1}} \|y\|^p\quad\text{for all }x,y\in X\text{ and }\lambda\in (0,1).
\]
Now, let $\lambda\in (0,1)$ such that $b:=\frac{c}{(1-\lambda)^{p-1}}<1$ and $|\mu|_p:=\int \|x\|^p\,\mu(\d x)$. Then,
\[
\int_X \Big(f\big(x+y(x)\big)-\|y(x)\|^p\Big)\, \mu(\d x)\leq M+\frac{c}{\lambda^{p-1}}|\mu|_p-(1-b)\int_X \|y(x)\|^p \, \mu(\d x)
\]
for all $y\in B(X,X)$.\ In particular, $f\in \cL$ and, using Corollary \ref{cor:selection},
\begin{equation}\label{eq.wasserstein}
\rho(f)=\sup_{y\in L^p(\mu;X)}\int_X \Big(f\big(x+y(x)\big) - \|y(x)\|^p\Big) \,\mu(\d y), 
\end{equation}
where $L^p(\mu;X)$ denotes the space of all ($\mu$-equivalence classes of) functions $y\in \Bm(X,X)$ with $\int_X\|y(x)\|^p\,\mu(\d x)<\infty$.

\subsection{Drift uncertainty on the path space} Let $X=Y=C_0([0,1])$ be the space of all continuous functions $x\colon [0,1]\to \R^d$ with $x(0)=0$, endowed with the topology induced by the supremum norm.\ Let $\mu\in \cP(X)$ be the Wiener measure, and consider the cost function $c\colon X\times X\to [0,\infty]$, given by
\[
c(x,y):=\begin{cases}\phi\big(\|x-y\|_{H^1}^2\big),& \text{if }x-y\in H^1\\
\infty, &\text{otherwise},\end{cases}
\]
with a nondecreasing function $\varphi\colon[0,\infty]\to[0,\infty]$. Here, 
$H^1=H^1([0,1])$ is the Cameron-Martin space, i.e., the space of all $u\in X$ with weak derivative $u':=L^2=L^2([0,1])$, and $\|u\|_{H^1}:=\|u'\|_{L^2}$ is the Cameron-Martin norm for $u\in H^1$.\ Note that, by definition of $X$, $u(0)=0$ for all $u\in X$, so that the \textit{first Poincar\'e inequality}
\[
\|u\|_\infty\leq \int_0^1|u'(t)|\, \d t\leq  \|u'\|_{L^2}=\|u\|_{H^1}
\]
is valid for all $u\in H^1$.\ In particular, the cost function $c$ is Borel-measurable. 

In this case, the cost function $C(x, \vartheta) := \int_X c(x,y)\,\vartheta(\d y)$ is of the form \eqref{eq.transport.cost} with $\cM=\cP(X)$ and $\tau_x=\id_X$ for all $x\in X$. Moreover, $\kappa_0(x):=\delta_x$ satisfies $C(x,\kappa_0(x))=0$ for all $x\in X$.\ The corresponding risk measure is given by
\begin{equation}\label{eq.foellmerex}
 \rho(f):=\sup_{\nu\in \cP(X)} \bigg(\int_X f(y)\, \nu(\d y)-\inf_{\pi\in \cpl(\mu,\nu)} \int_{X\times X}c(x,y)\, \pi(\d x,\d y)\bigg) \
\end{equation}
for all $f\in \Bm(X)$.\ By Proposition~\ref{prop.reduction.extr}, it follows that
\begin{align*}
\rho(f)&=\int_X \sup_{u\in H^1} f(x+u)-\phi\big(\|u\|_{H^1}^2\big)\,\overline\mu(\d x)\\
&=\int_X \sup_{\theta\in L^2} \Bigg(f\bigg(x+\int_0^\cdot \theta_s\, \d s\bigg)-\phi\bigg(\int_0 ^1|\theta_s|^2\, \d s\bigg)\Bigg)\,\overline\mu(\d x) \quad\text{for all }f\in \cL.
\end{align*}
If $\phi(v)=\frac{v}{2}$, by \cite[Proposition 1]{foellmer2022optimal} or \cite[Proposition 2]{foellmer2022optimal}, it follows that
\begin{equation}\label{eq.foellmer}
\rho (f)\geq \sup_{\nu\in \cP(X)} \bigg(\int_X f(x)\, \nu(\d x)- H(\nu|\mu)\bigg)=\log\bigg(\int_X e^{f(x)}\, \mu(\d x)\bigg)
\end{equation}
for all $f\in \cL$, where $H(\nu|\mu)$ is the relative entropy of $\nu\in \cP(X)$ with respect to $\mu$.\ Note that, in general, equality in \eqref{eq.foellmer} only holds if the penalty term in \eqref{eq.foellmerex} is replaced by the Cameron-Martin adapted Wasserstein distance, see \cite[Corollary 2]{foellmer2022optimal} for the details.

\subsection{Martingale constraints} \label{sec:mart.constr}
We consider the case, where $X=Y=\R^d$ and the ambiguity set is given by all measures that satisfy a martingale constraint with respect to the reference measure $\mu\in \cP(\Rd)$.\ More precisely, we only consider probability measures $\nu \in \cP(X)$ that admit a martingale coupling with the reference measure $\mu$, i.e., a coupling $\pi \in \cpl(\mu, \nu)$ with
\begin{equation} \label{eq.mart.coupl}
\int_{\Rd \times \Rd} (y-x) \cdot \one_B(x)\,\pi(\d x, \d y) = 0 \quad \text{for all } B \in \cB (X).
\end{equation}
By Strassen \cite{strassen1965existence}, the existence of a martingale coupling between two measures $\mu$ and $\nu$ is equivalent to $\mu$ and $\nu$ being in convex order. We consider the risk measure
\begin{equation} \label{eq.static.mart}
\rho^\mart(f) := \sup_{\mu \in \cP(\R^d)} \bigg( \int_{\R^d} f(y)\,\nu(\d y) - \inf_{\pi \in \mart (\mu, \nu)} \int_{\R^d \times \R^d} c(x,y)\,\pi(\d x, \d y) \bigg)
\end{equation}
for $f\in \Bm(\R^d)$, where $c \colon [0,\infty) \to [0, \infty]$ is a Borel measurable cost function and $\mart(\mu, \nu)$ denotes the set of martingale couplings between $\mu$ and $\nu$.

Then, we can express this risk measure via a weak optimal transport penalty with cost function
\[
C(x,\vartheta) = \begin{cases}
    \int_{\R^d} c(x,y)\,\vartheta(\d y), &\text{if } \vartheta\tau_x^{-1} \in \cM_0, \\
    \infty, & \text{otherwise,}
\end{cases}
\]
with
$$
\cM_0 = \bigg\{ \vartheta \in \cP(\R^d)\colon \int_{\R^d} x\,\vartheta(\d x) = 0 \bigg\}
$$
and $\tau_x(y):=x+y$ for all $x,y\in \R^d$, i.e., $\tau_x^{-1}(y)=y-x$ for all $x,y\in \R^d$.

In fact, let $\nu\in \cP(\R^d)$, $\pi \in \mart(\mu, \nu)$, and $\kappa(x):=\pi(\, \cdot\, |x)$ for all $x\in X$.\ Then, $\kappa \in \ker(\mu, \nu)$ and $\kappa(x)\tau_x^{-1}\in \cM_0$.\ On the other hand, every kernel $\kappa \in \ker(\mu, \nu)$ with $\kappa(x)\tau_x^{-1} \in \cM_0$ gives rise to a coupling $\pi \in \mart(\mu, \nu)$ via $\pi (A\times B):= \int_A \kappa(x,B)\,\mu(\d x)$ for all $A,B\in \cB(\R^d)$. Hence,
\[
\inf_{\kappa \in \ker(\mu, \nu)} \int_{\R^d} C(x, \kappa(x))\,\mu(\d x) = \inf_{\pi \in \mart(\mu, \nu)} \int_{\Rd \times \Rd} c(x,y)\,\pi(\d x, \d y),
\]
and therefore $\rho^\mart = \rho$.

Observe that the set $\cM_0$ is defined by a single moment constraint, namely a mean constraint.\ Hence, Theorem \ref{thm:winkler} implies that
\[
\ex \cM_0 = \bigg\{ \nu \in \cP(\Rd) \,:\, \nu = p\delta_{x} + (1-p)\delta_{-\frac{p}{1-p}x},\, x \in X,\, p \in (0,1) \bigg\}.
\]
Therefore, by Proposition \ref{prop:selection}, the martingale constraint reduces to the following optimization
\begin{equation} \label{eq.mart.reduced}
\begin{split}
\rho(f) = \int_{\Rd} \sup_{\substack{y \in X \\ p \in (0,1)}} & \Bigg(p \big(f(x + y) - c(|y|)\big) \\
    & + (1-p)\bigg(f\Big(x - \tfrac{p}{1-p} y \Big) - c\Big(\tfrac{p}{1-p} |y|\Big) \bigg)\Bigg)\,\overline\mu(\d x).
\end{split}
\end{equation}
This means that the martingale constraint, which is an infinite dimensional constraint, on the state space $\R^d$ reduces to a $d+1$-dimensional optimization problem.\ In Section \ref{sec.numerics}, we will use this representation of the martingale constraint to compute no-arbitrage bounds for prices of financial derivatives after quoted maturities.

\subsection{Martingales on the Skorokhod space}
Let $X=Y=\R^d$, and consider the space $\Omega:=D=D([0,1])$ of all c\`adl\`ag functions $[0,1]\to \R^d$, endowed with the $\sigma$-algebra $\cF:=\sigma(\{\pr_t\colon t\in [0,1]\})$ generated by the canonical process $(\pr_t)_{t\in [0,1]}$.\ Let $\mu\in \cP(\R^d)$ with $\int_{\R^d}|x|^2\, \mu(\d x)<\infty$. Then, we denote by $\mart_D^2(\mu)$ the set of all martingale measures $\eta$ on $(\Omega,\cF)$ with $\eta\circ \pr_0^{-1}=\mu$ and $\int_\Omega |\omega_1|^2\,\eta(\d \omega)<\infty$. Consider the risk measure
\[
\rho(f):=\sup_{\eta\in \mart_D^2(\mu)} \int_\Omega f(\omega_1)-\langle \omega\rangle_1\, \eta(\d \omega),
\]
where $\langle \,\cdot\,\rangle$ denotes the quadratic variation.\ Then, by It\^o's formula,
\[
\int_\Omega \langle \omega\rangle_1\, \eta(\d \omega)=\int_{\Omega} |\omega_1|^2\,\eta(\d \omega)-\int_{\R^d}|x|^2\,\mu(\d x)
\]
for all $\eta\in \mart_D^2(\mu)$. By \cite[Remark 2.2 (ii)]{MR3605716},
we thus find that
\[
\rho(f)=\sup_{\nu\in \cP(\R^d)} \int_{\R^d}f(y)\,\nu(\d y)-\cW_2^\mart(\mu,\nu)^2.
\]
As in Section \ref{sec:mart.constr}, we can therefore apply Theorem \ref{thm.main}, and obtain that
\[
\rho(f)=\int_{\R^d}\sup_{\theta\in \R^d}\sup_{p\in (0,1)} \Big(pf\big(x+\tfrac{\theta}{p}\big)+(1-p)f\big(x-\tfrac{\theta}{1-p}\big)- \tfrac{|\theta|^2}{p(1-p)}\Big)\,\mu(\d x).
\]

\subsection{Matching quoted option prices on a given maturity} \label{sec:model.free.pricing}
Let $X=Y=\R^d$ with $d \in \N$, and consider a two-period market consisting of $d$ assets.\ We assume that the distribution of the assets today, i.e., at time $t=0$, is known and given by a Dirac measure $\mu = \delta_{x_0}$ with $x_0 \in \R^d$.\ We assume that the discounted prices of the assets at time $t=1$ are given by a random vector $x=(x_j)_{j=1,\ldots,d}$, whose distribution is not known precisely.

Let $f_1 \dots, f_n \colon \R \to [0, \infty)$ be the payoffs of $n \in \N$ quoted options, e.g., call options $f_i(x) = (x_i - K_i)^+$ with $K_i \in [0, \infty)$ for $i=1,\ldots,n$.\ The market then gives a set of intervals $I_i=[b_i, a_i]$ with $b_i, a_i \in [0, \infty)$ representing the bid-ask spread for the option $f_i$ for $i=1,\ldots, n$.\ If $\nu$ is the law of $x$, it should satisfy $\int_{\R^d} f_i(y)\,\nu(\d y) \in [b_i,a_i]$, for all $i =1, \dots, n$.\ This conditions can be represented as a set of $2n$ generalized moment constraints. In fact, for each $i$ we consider the constraints $$\int_{\R^d} f_i(y)\,\nu(\d y) \le a_i\quad\text{and}\quad\int_{\R^d} - f_i(y)\,\nu(\d y) \le - b_i.$$
Moreover, we require that the measure $\nu$ does not generate an arbitrage opportunity by imposing the constraint $\int_X y\,\nu(\d y) = x_0$ on the mean of the distribution at time $t=1$. 

We can then express the upper no-arbitrage price bound $\overline V(f)$ for a European option with payoff $f \colon \R^d \to \R$ and maturity $t=1$ using the reference measure $\mu = \delta_{x_0}$ and the cost function $C(x, \vartheta) = \infty \cdot \one_{\cM^c}(\vartheta)$ for $x\in \R^d$ and $\vartheta\in \cP(\R^d)$ with $\cM$ being the set of all probability measures that satisfy the $2n+2$ constraints described above. We may then choose the kernel $\kappa_0$ arbitrarily among all kernels in $\ker(\R^d,\R^d)$ with $\kappa_0(x_0)\in \cM$, and consider
\begin{align*}
\overline V(f) & := \sup_{\nu \in \cM} \int_X f(y)\,\nu(\d y) \\
 & = \sup_{\nu \in \cP(\R^d)} \bigg(\int_X f(y)\,\nu(\d y) - \inf_{\kappa \in \ker(\mu, \nu)} C\big(x_0, \kappa(x_0)\big) \bigg) = \rho(f).
\end{align*}
Here, the second and third equality directly follow from the choice of the reference measure $\mu = \delta_{x_0}$, which implies that $\kappa(x_0) = \nu \in \cM$ for all $\kappa \in \ker(\mu, \nu)$ with $C\big(x_0,\kappa(x_0)\big)<\infty$.

Using the IRP of the set $\cM$ and Theorem \ref{thm.reduction.extr}, together with Theorem \ref{thm:winkler}, we can solve this problem by optimizing over $2n+3$ Dirac measures.\ This leads to
\[
\rho(f) = \sup_{(p,y) \in M} \sum_{i=1}^{2n+3} p_i f(y_i),
\]
where $M = \big\{ (p,y) \in [0,1]^{2n+3} \times (\R^d)^{2n+3} \colon \sum_{i=1}^{2n+3} p_i \delta_{y_i} \in \cM \big\}$.

In a similar way, we can obtain the lower price bound by considering $\underline{V}(f) := -\rho(-f)$, see Section \ref{sec:pricing.after} below.

\section{Numerics} \label{sec.numerics}

In this section, we complement our main result by additional approximation results, which are then applied in a series of numerical examples.

\subsection{Wasserstein uncertainty without constraints}
Let $X=Y$ be a separable Banach space.\ We start by considering the risk measure from Section \ref{sec:wass.unc}, i.e.,
\begin{equation} \label{eq.wass.unconstr}
\rho(f) = \sup_{\nu \in \cP(X)} \bigg( \int_X f(y)\,\nu(\d y) - \cW_p(\mu, \nu)^p \bigg)\quad\text{for }f\in \Bm(X)
\end{equation}
with $p \ge 1$ and $\mu \in \cP_p(X)$. Recall that the risk measure \eqref{eq.wass.unconstr} corresponds to the weak transport cost $C(x, \vartheta):= \int_X |y - x|^p\,\vartheta(\d y)$ for $x\in X$ and $\vartheta\in \cP(X)$ and that, by \eqref{eq.wasserstein}, it can be expressed as
\[
\rho(f) = \sup_{y \in L^p(\mu;X)} \int_X \Big(f\big(x+y(x)\big) - \|y(x)\|^p\Big)\,\mu(\d x) \quad\text{for }f\in \Bm(X)\text{ with \eqref{eq.growthf}}.
\]

The next proposition provides an approximation from below for this risk measure, based on an optimization over a set which is dense in $L^p(\mu;X)$.

\begin{proposition}\label{prop.dense.approx}
    Let $f \colon X\to \R$ be continuous with \eqref{eq.growthf}, and let $(D_n)_{n \in \N}$ be a family of subsets of $L^p(\mu;X)$ with $D^n \subset D^{n+1}$ for all $n \in \N$. If $D:=\bigcup_{n \in \N} D_n$ is dense in $L^p(\mu;X)$, then the risk measure \eqref{eq.wass.unconstr} can be approximated from below in the following sense:
    \[
    \sup_{y \in D^n} \int_X \Big(f\big(x + y(x)\big) - \|y(x)\|^p\Big)\,\mu(\d x) \nearrow \rho(f) \quad \text{as } n \to \infty.
    \]
\end{proposition}

\begin{proof}
    Since $D \subset L_p(\mu;X)$, it follows that
    \begin{align*}
    \rho(f)&\geq \sup_{y \in D} \int_X \Big(f\big(x + y(x)\big) - \|y(x)\|^p\Big)\,\mu(\d x).
    \end{align*}
    To prove the other inequality, let $\eps > 0$ and $y_0 \in L_p(\mu;X)$ with
    \[
    \sup_{y \in L_p(\mu;X)} \int_X \Big(f\big(x + y(x)\big) - \|y(x)\|^p\Big)\,\mu(\d x) \le \int_X \Big(f\big(x + y_0(x)\big) - \|y_0(x)\|^p\Big)\,\mu(\d x) + \frac{\eps}{2}.
    \]
    For $y \in L^p(\mu;X)$, consider the probability measure $\mu_y:=\mu \circ \big(x \mapsto x + y(x)\big)^{-1}\in \cP_p(X)$. Then, using the comonotone coupling,
    $$
    \cW_p(\mu_{y_0}, \mu_{y}) \le \|y_0 - y\|_{L^p(\mu;X)}\quad \text{for all }y\in L^p(\mu;X).
    $$
    Since $D$ is dense in $L^p(\mu;X)$ and the Wasserstein distance metrizes the weak topology on $\cP_p(X)$, cf. \cite[Definition 8.8 and Theorem 6.9]{villani2008optimal}, there exists some $y_D \in L^p(\mu;X)$ with
    \[
    \int_X \Big(f\big(x + y_0(x)\big) - \|y_0(x)\|^p\Big)\,\mu(\d x) \le \int_X \Big(f\big(x + y_D(x)\big) - \|y_D(x)\|^p\Big)\,\mu(\d x) + \frac{\eps}{2}.
    \]
     Altogether, we find that
    \[
    \sup_{y \in L_p(\mu;X)} \int_X \Big(f\big(x + y(x)\big) - \|y(x)\|^p\Big)\,\mu(\d x) \le \int_X \Big(f\big(x + y_D(x)\big) - \|y_D(x)\|^p\Big)\,\mu(\d x) + \eps.
    \]
    Letting $\eps \downarrow 0$ and taking the supremum over all $y\in D$, we obtain that
    \begin{equation} \label{eq:dense.approx}
    \rho(f) = \sup_{y \in D} \int_X \Big(f\big(x + y(x)\big) - \|y(x)\|^p\Big)\,\mu(\d x).
    \end{equation}
    Since $D^n \subseteq D^{n+1}$ for all $n \in \N$,
    \[
        \sup_{y \in D} \int_X \Big(f\big(x + y(x)\big) - \|y(x)\|^p\Big)\,\mu(\d x) = \sup_{n\in \N}\sup_{y \in D_n} \int_X \Big(f\big(x + y(x)\big) - \|y(x)\|^p\Big)\,\mu(\d x),
    \]
    and the claim follows from equation \eqref{eq:dense.approx}.
\end{proof}

The previous proposition provides a constructive way of approximating the risk measure \eqref{eq.wass.unconstr}:\ given an increasing sequence of subsets of $L_p(\mu;X)$, whose union is dense in $L^p(\mu;X)$, we can approximate the risk measure from below by solving a variational problem on the increasing sequence of subsets.\ From a practical point of view, the previous result allows to approximate the $C$-transform of the loss function by looking at an optimizing vector field instead of solving a pointwise optimization.

Following closely the approach in \cite[Section 3]{nendel2022parametric}, we can apply this result in the case $X = \R^d$ with $d\in \N$ using the numerical framework of neural networks. For the sake of a self contained exposition, we provide a full description of the framework and restate the approximation result for this explicit choice of increasing subsets.\ We denote by $\fN_d^{m,n}$ the set of fully connected feed-forward neural networks from $\R^d$ to $\R^d$ with $m$ hidden layers and $n$ neurons per layer, i.e., the set of functions
\[
x \mapsto A_m \circ \psi \circ A_{m - 1} \circ \dots \circ \psi \circ A_0(x),
\]
where $A_i$ are affine transformations of the form $A_i = M_i x + b_i$ with a matrix $M_i$ and a vector $b_i$, $\psi \colon \R \to \R$ is a nonlinear function, called the \textit{activation function}, and the composition with the activation function is to be understood componentwise, i.e., $$\psi((x_1, \dots, x_d)) = (\psi(x_1), \dots, \psi(x_d)).$$ For a neural network in $\fN_d^{m,n}$, $A_0$ is a function from $\R^d$ to $\R^n$, $A_1, \ldots ,A_{m-1}$ are functions from $\R^n$ to $\R^n$, and $A_m$ is a function from $\R_n$ to $\R_d$.\ The fact that $\fN_d^{m,n}$ is dense in $L^p(\mu;\R^d)$ for suitable activation functions allows us to state the following corollary.\ Its proof  relies on classical universal approximation results, cf.\ \cite{hornik1991approximation,hornik1989universal}, and follows closely the one of \cite[Corollary 3.3]{nendel2022parametric}.

\begin{corollary} \label{cor.nn.approx}
    For every $m \in \N$ and every nonconstant Lipschitz continuous activation function $\psi$, 
    \[
    \sup_{y \in \fN_d^{m,n}} \int_{\R^d} \Big(f\big(x + y(x)\big) - |y(x)|^p\Big)\,\mu(\d x) \nearrow \rho(f) \quad \text{as } n \to \infty
    \]
for all continuous functions $f\colon \R^d\to \R$ with \eqref{eq.growthf}.
\end{corollary}

\begin{proof}
    Since the activation function $\psi$ is Lipschitz continuous, the functions in $\fN_d^{m,n}$ are also Lipschitz continuous as a composition of Lipschitz continuous functions, so that $\fN_d^{m,n} \subset L^p(\mu;X)$.\ For $m = 1$, the convergence follows directly from Proposition \ref{prop.dense.approx} and \cite[Theorem 1]{hornik1991approximation}.  Moreover, $\fN_d^{1,n} \subseteq \fN_d^{m,n}$, which proves the approximation for all $m\in \N$.
	\end{proof}

In this framework, the optimization in the space of neural networks can be performed via gradient descent methods minimizing the loss function
\[
    L(y) = \int_{\R^d} \Big(f\big(x + y(x)\big) - |y(x)|^p\Big) \,\mu(\d x),
\]
where the integral is computed numerically, using either quadrature formulas or a Monte Carlo simulation. Note that, when computing the integral with a Monte Carlo simulation, new samples are drawn at each step of the training, so that gradient descent automatically results in a stochastic gradient descent algorithm.

\subsubsection{A two dimensional example with Wasserstein-2 penalization}
We provide an example for this framework recalling the earthquake example from \cite{nendel2022parametric}.\ We consider a one-period model, where $X = \R^2$ and the reference measure $\mu$ describes an a priori estimate for the location of the epicenter of an earthquake.\ Suppose an insurance company, selling coverage for the damages caused by earthquakes, estimates a loss function that depends on the location of the epicenter of the earthquake and computes a worst case loss penalizing deviations from the reference model with the Wasserstein-2 distance. We thus consider the risk measure \eqref{eq.wass.unconstr} with $p=2$, that is
\[
\rho(f) = \sup_{\nu \in \cP_2(x)} \bigg( \int_{\R^2} f(x)\,\nu(\d x) - \cW_2(\mu, \nu)^2 \bigg).
\]
Using Theorem \ref{thm.main}, one my compute the risk measure $\rho$ by estimating the function $f^C$ and computing its expectation with respect to the reference measure $\mu$, e.g., via a Monte Carlo simulation. We work with the loss function depicted in Figure \ref{fig.earthquake}\subref{subfig:earth.loss}, which models the dependence of the losses on the density of the buildings in a city, or on the distance from two city centers with different exposures.\ As a reference measure $\mu$, we choose the law of a two dimensional Gaussian random variable with mean located at $(0.75, 0.25)$ and identity matrix as covariance matrix. 
We estimate this risk measure by approximating $f^C$ with the neural network approach described above.

\begin{figure}
    \subfloat[Loss function $f$\label{subfig:earth.loss}]{
    \includegraphics[width=0.45\textwidth]{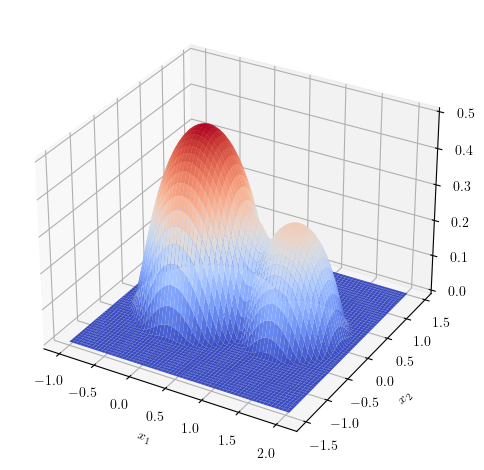}
    }
    \hfill
    \subfloat[$C$-transform of $f$\label{subfig:earth.ctrans}]{
    \includegraphics[width=0.45\textwidth]{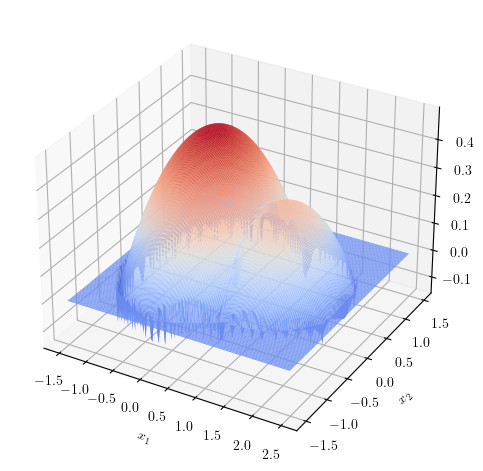}
    }
    \caption{On the left, loss function for the example with Wasserstein uncertainty without constraints. On the right, corresponding $C$-transform.}
    \label{fig.earthquake}
\end{figure}

\begin{figure}
    \centering
    \includegraphics[width=0.5\textwidth]{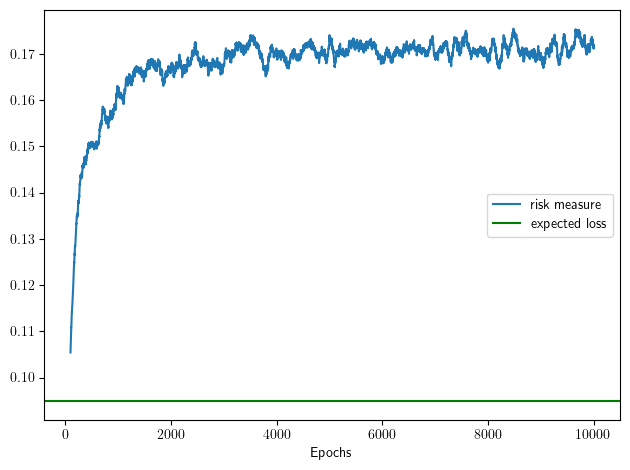}
    \caption{Training of the risk measure for the two dimensional example with Wasserstein penalization.\ The values plotted are moving averages with rolling window of 100 epochs.}
    \label{fig.2dimtrain}
\end{figure}

We use a neural network consisting of $4$ hidden layers and $20$ neurons per layer with ReLU activation function, i.e., $\psi(x) = \max\{x, 0\}$.
The optimization of the neural network is obtained using the Adam optimizer, cf.\ \cite{adam2014}, with a learning rate of $0.001$ and a batch size of $100$ samples, which means that at each iteration only $100$ sample points are used to compute the Monte Carlo.\footnote{All numerical analyses were performed on a laptop equipped with a 4-Core 3.48GHz CPU and 16GB RAM. The examples were implemented in Python 3.9 and the neural network architecture was handled using the package PyTorch (version 1.11.0). Source codes are available at \url{https://github.com/sgarale/risk_meas_wot}.}
Figure \ref{fig.2dimtrain} shows the training phase of the risk measure with the neural network approach and Figure \ref{fig.earthquake}\subref{subfig:earth.ctrans} depicts the resulting $C$-transform after $10000$ training epochs.
In order to obtain a numerical value for $\rho(f)$, one only has to perform a final Monte Carlo simulation to integrate the resulting $C$-transform $f^C$ with respect to the reference measure $\mu$, and this last evaluation can be performed at a low computational cost with a large sample, e.g., one million points.

\subsection{Martingale constraint:\ pricing after quoted maturities} \label{sec:pricing.after}

In this section, we apply the framework from Section \ref{sec:mart.constr} to the pricing of derivatives in a model-free setting. We assume that $X=Y=\Rd$ with $d\in \N$, and consider an arbitrage-free market with $d$ assets and risk-free rate zero.\ Suppose that there are many options quoted on a certain maturity $T>0$, and we therefore know the distribution of the assets at time $T$, and that there are no more quoted options after this maturity.\ If we want to price a European derivative with maturity $T + t$, we then have no information for the calibration of a model at that maturity. Using, however, a model-free approach, we can compute bounds on the option price looking at all probability measures which dominate the distribution of the underlying at time $T$ in convex order. In fact, even without making rigorous assumptions on the market model, we can assume that absence of arbitrage implies that (discounted) asset processes are martingales.

As a best guess for the distribution at time $T+t$, we choose the distribution $\mu$ of the assets at time $T$, and we penalize deviations from $\mu$ by means of a martingale optimal transport cost as discussed in Section \ref{sec:mart.constr}.

In order to take into account the length of the time step $t$, we introduce a parameter that controls the level of uncertainty by a rescaling of the penalty term, considering the cost $c_t(v) := t c\big(\frac{v}{\sqrt{t}}\big)$ for $t>0$ and $v \ge 0$. Roughly speaking, the assumption underlying this rescaling is that the uncertainty in the distribution grows proportionally to the square root of time.

We denote by $\overline{V}_{T+t}(f)$ and $\underline{V}_{T+t}(f)$ the upper and lower bound for the option price at time $T+t$, respectively, i.e.,
\[
\overline{V}_{T+t}(f) = \rho_t^\mart(f) := \sup_{\mu \preceq \nu} \bigg( \int_X f(y)\,\nu(\d y) - \inf_{\pi \in \mart (\mu, \nu)} \int_{X \times X} c_t(|y - x|)\,\pi(\d x, \d y) \bigg).
\]
Note that $\rho_t^\mart(f)$ corresponds to equation \eqref{eq.static.mart} with an additional time scaling parameter.

Using the same arguments, it is natural to define the lower bound for the option price considering the infimum over probability measures that dominate $\mu$ in convex order and penalize them with the same cost. This leads to
\[
\underline{V}_{T+t}(f) := \inf_{\mu \preceq \nu} \bigg( \int_X f(y)\,\nu(\d y) + \inf_{\pi \in \mart (\mu, \nu)} \int_{X \times X} c_t(|y - x|)\,\pi(\d x, \d y) \bigg) = -\rho_t^\mart(-f).
\]

Numerically, we will solve the reduced problem \eqref{eq.mart.reduced} with $c(v) = |v|^p$ and with the additional time parameter in the cost function. As in the previous example, we solve this problem using Proposition \ref{prop:selection}, optimizing outside of the integral and reducing the optimization to the space of neural networks with given width and depth, i.e., we consider the optimization
\[
\begin{split}
    \sup_{(y,v) \in \fN_{d,d+1}^{m,n}} & \int_X \bigg(\phi\big(v(x)\big) \Big(f\big(x + y(x)\big) - c\big(|y(x)|\big)\Big) \\
    & + \Big(1-\phi\big(v(x)\big)\Big) \bigg[f\bigg(x - \frac{\phi\big(v(x)\big)}{1-\phi\big(v(x)\big)} y(x) \bigg) - c_t\bigg(\frac{\phi\big(v(x)\big)}{1-\phi\big(v(x)\big)} |y(x)|\bigg)\bigg]\,\mu(\d x),
\end{split}
\]
where $(y,v) \in \fN_{d,d+1}^{m,n}$, indicates a neural network from $\Rd$ to $\R^{d+1}$, and we denote its components by $y \colon \Rd \to \Rd$ and $v \colon \Rd \to \R$.\ Moreover, we use the transform $\phi(v) = \frac{1}{1 + e^{-v}}$ to map real numbers to probabilities.

\subsubsection{Bull spread option}
We consider a bull call spread option, i.e., a European option with payoff given by the sum of a long call option with strike $K_1 > 0$ and a short call option with strike $K_2 > K_1$.\ This leads to the payoff function
\[
f(x) = (x - K_1)^+ - (x - K_2)^+.
\]
As a reference model, we consider the log-normal distribution corresponding to the marginal at time $T = 0.5$ of a Black-Scholes model with volatility $\sigma = 0.2$.\ We set the initial value of the asset to one and the strikes to $K_1 = 0.9$ and $K_2 = 1.2$.

We work with the cost function $c(v)=|v|^3$, which corresponds to a Wasserstein-3 penalization and compute the upper bound $\overline{V}_{T+t}(f)$ and the lower bound $\underline{V}_{T+t}(f)$ for the arbitrage-free price of the option at different maturities after the quoted maturity $T$. We compare these bounds with the price given by the reference model, i.e., the Black-Scholes price.\ The output is displayed in Figure \ref{fig:bull.spread}\subref{subfig:price.bounds}.\ In Figure \ref{fig:bull.spread}\subref{subfig:bull.spread}, we depict the payoff function and the resulting $C$-transform for a specific maturity.

\begin{figure}
    \subfloat[Payoff and its $C$-transform\label{subfig:bull.spread}]{
    \includegraphics[width=0.45\textwidth]{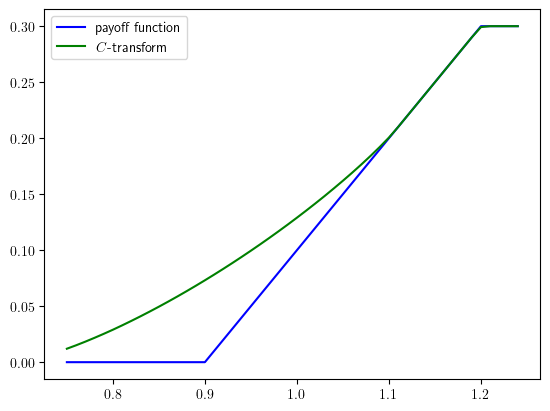}
    }
    \hfill
    \subfloat[Upper and lower price bounds\label{subfig:price.bounds}]{
    \includegraphics[width=0.45\textwidth]{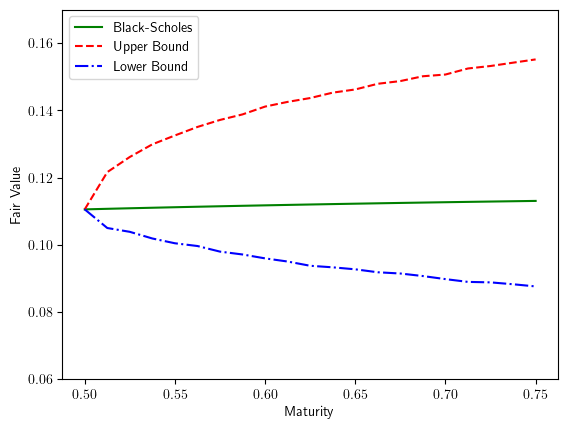}
    }
    \caption{On the left, payoff function and $C$-transform for the upper bound price of the bull spread option $1$ month after the quoted maturity. On the right, upper and lower bounds for the price of the option at different times after the quoted maturity.\ Quoted maturity:\ $T=0.5$.\ Reference model:\ Black-Scholes model with volatility $\sigma=0.2$.}
    \label{fig:bull.spread}
\end{figure}

\subsubsection{Max call option on two assets}
We extend the previous example to a market with two assets.\ We compute the upper bound for the price of a max call option, which is the European option defined by the payoff
\[
f(x_1,x_2) = \big(\max\{x_1,x_2\} - K)^+,
\]
with $K > 0$.\ As reference measure, we consider the marginal at time $T=1$ of a diffusive model composed by two assets whose dynamics under the equivalent martingale measure $\Q$ is given by the SDE
\[
\begin{cases}
    \d X_t = \sigma \d B_t, \\
    X_0 = x_0,
\end{cases}
\text{with } X_t = \begin{pmatrix} X_t^1 \\ X_t^2\end{pmatrix},\,
x_0 = \begin{pmatrix} 1\\1 \end{pmatrix},\,
B_t = \begin{pmatrix} B_t^1 \\ B_t^2\end{pmatrix},\,
\sigma = \begin{pmatrix} 0.30 & 0 \\ 0.05 & 0.20 \end{pmatrix},
\]
where $(B_t)_t$ is a standard two dimensional Brownian motion under $\Q$, see \cite[Chapter 1]{karatzas1998finance}.
We fix $t=1/12$ and penalize deviations from the reference measure with a Wasserstein-3 penalty.
As before, we can employ the previously described neural network approach to compute the upper bound for the price of the max call option with strike at the money $K = 1$. Figure \ref{fig.max.call} displays the payoff function and its $C$-transform.

\begin{figure}
    \subfloat[Neural network optimization]{
    \includegraphics[width=0.45\textwidth]{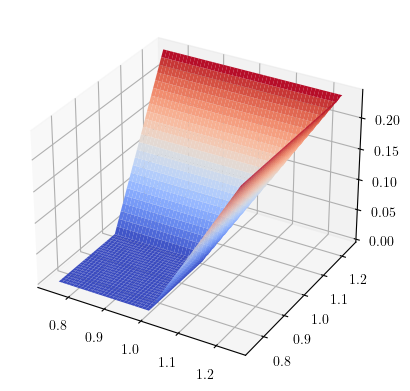}
    }
    \hfill
    \subfloat[Pointwise optimization]{
    \includegraphics[width=0.45\textwidth]{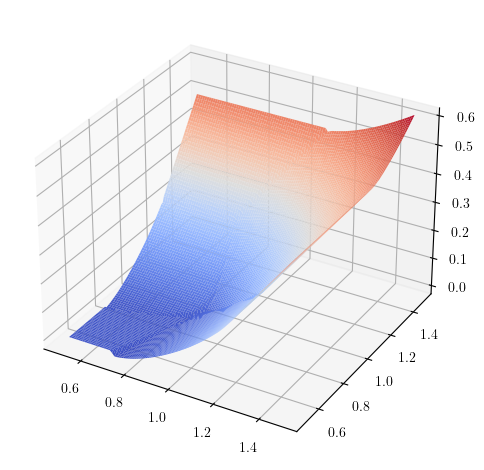}
    }
    \caption{Payoff function and $C$-transform for the upper bound price of the max call option at maturity $1$ year and $1$ month. On the left, payoff function, on the right, $C$-transform. Strike level $K = 1$.}
    \label{fig.max.call}
\end{figure}

\subsubsection{Higher dimensions}\label{sec:higherdim}
We can also work in markets with a higher number of assets. As a matter of fact, the training of the neural network is based on sampling from the reference measure and therefore avoids the \textit{curse of dimensionality}, which arises in methods based on the discretization of the underlying space.\ As the number of assets in the market increases, the network reproducing the measurable selection of the optimizers for the $C$-transform of the payoff function becomes larger, but this does not dramatically affect the performance of the method.\ In Figure \ref{fig:multidim}, we display the result of the evaluation of the upper bound for the price of several options as the number of underlying assets increases, compared with their evaluation time (in seconds).\ For each value, we train the neural network for $10^4$ epochs.\ Again, the reference model is a diffusive model in which each asset has volatility $0.2$ and is driven by a standard Brownian motion.\ Moreover, we assume that the assets start  at time zero from the value $1$.\ Clearly, as in the previous section, it is possible to impose any correlation structure between the assets.\ We assume that the quoted maturity is $T=1$ (one year) and we compute the upper bound for the price at time $T+t$ with $t=1/12$, i.e., one month after the quoted maturity.

\begin{figure}
    \subfloat[Max call option]{
    \includegraphics[width=0.45\textwidth]{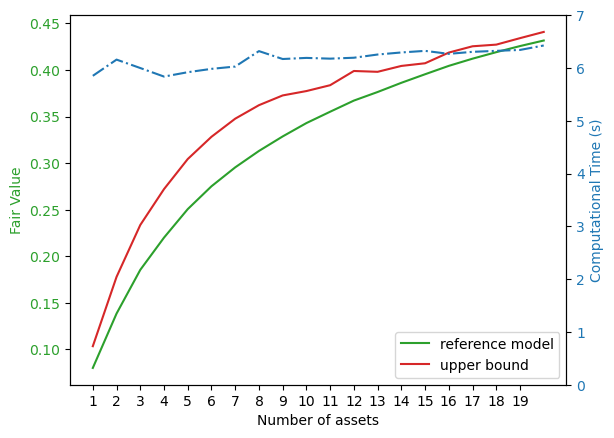}
    }
    \hfill
    \subfloat[Basket call option]{
    \includegraphics[width=0.45\textwidth]{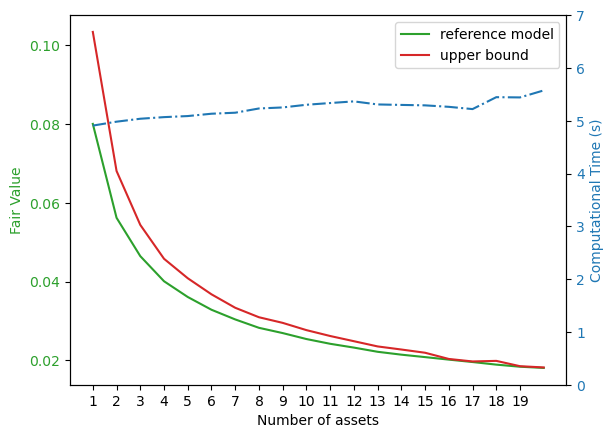}
    } \\
    \subfloat[Min put option]{
    \includegraphics[width=0.45\textwidth]{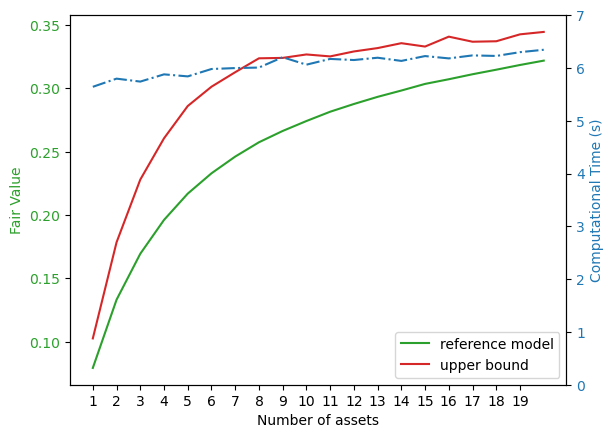}
    }
    \hfill
    \subfloat[Geometric put option]{
    \includegraphics[width=0.45\textwidth]{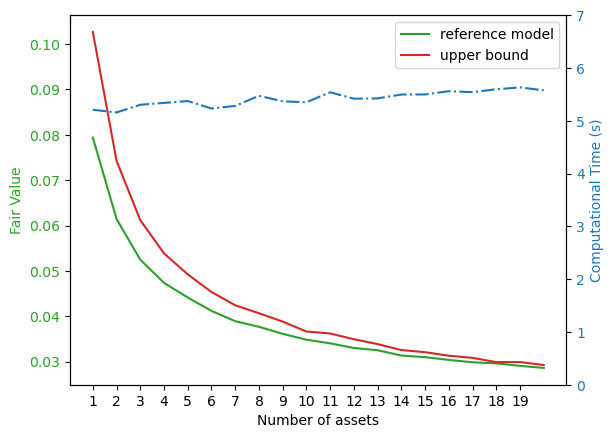}
    }
    \caption{Option prices and evaluation time in seconds as the number of assets in the market increases. Strike of the options $K=1$.}
    \label{fig:multidim}
\end{figure}

We consider four different options, which have the following payoffs:
\begin{align*}
    \text{max call:} \qquad &  f(x_1, \dots, x_d) = \big(\max\{x_1,\dots,x_d\} - K\big)^+ \\
    \text{basket call:} \qquad& f(x_1, \dots, x_d) = \bigg( \frac{1}{d}\sum_{i=1}^d x_i - K\bigg)^+ \\
    \text{min put:}   \qquad &  f(x_1, \dots, x_d) = \big(K - \min\{x_1,\dots,x_d\}\big)^+ \\
    \text{geometric put:} \qquad & f(x_1, \dots, x_d) = \bigg(K - \bigg( \prod_{i=1}^d x_i \bigg)^\frac{1}{d}\bigg)^+
\end{align*}
These options depend on a single strike $K > 0$, which, in the numerical examples, is set to be at the money $K=1$.


\end{document}